\documentclass[12pt]{amsart}
\usepackage{amsmath,amssymb, mathtools}
\usepackage[letterpaper]{geometry}
\usepackage[pdftex,bookmarks=true,bookmarksnumbered=true]{hyperref}
\hypersetup{
pdfauthor = {Jeffrey H Schenker},
pdftitle = {Estimating comples eigenvalues of a non-self adjoint Schoedinger operator via complex dilations},pdfcreator = {LaTeX with hyperref package},
pdfproducer = {pdftex}}

\usepackage{mathpazo}
\usepackage{helvet}
\usepackage{courier}
\usepackage{url}
\usepackage[T1]{fontenc}
\usepackage[latin9]{inputenc}

\usepackage{amsopn}
\usepackage{graphicx}
\usepackage{fixltx2e}
\usepackage{url}
\usepackage{ulem}

\numberwithin{equation}{section} 
\numberwithin{figure}{section} 

\theoremstyle{plain}
\newtheorem{thm}{Theorem}[section]
\newtheorem{lem}[thm]{Lemma} 

\theoremstyle{definition}

\theoremstyle{remark}
 
\newtheorem*{rems}{Remarks}

\newcommand{\ip}[1]{\left\langle #1 \right \rangle}
\newcommand{\norm}[1]{\left\Vert#1\right\Vert}
\newcommand{\abs}[1]{\left\vert#1\right\vert}
\newcommand{\set}[1]{\left\{#1\right\}}

\newcommand{\cu}[1]{\mathcal{#1}}

\def\e{\mathrm e}

\def\im{\mathrm i}
\def\Im{\mathrm{Im}}

\def\1{{\mathsf 1}}
\def\di{\mathrm d}

\def\R{\mathbb R}
\def\C{\mathbb C}

\def\ra{\rightarrow}

\def\ran{\operatorname{ran}}

\def\dim{\operatorname{dim}}  

\def\num{\operatorname{num}}
\def\cl{\operatorname{cl}}
\def\cch{\operatorname{cch}}

\def\Re{\operatorname{Re}}
\def\Im{\operatorname{Im}}

\begin{document}

\title[Estimating complex spectrum via complex dilations]{Estimating 
complex eigenvalues of non-self adjoint Schr\"odinger operators via complex dilations}%
\author{Jeffrey H. Schenker}%
\address{Department of Mathematics, Michigan State University, East Lansing MI 48823, USA}%
\email{jeffrey@math.msu.edu}%


\begin{abstract}
The phenomenon ``hypo-coercivity,'' i.e., the increased rate of contraction for a semi-group upon adding a large skew-adjoint part to the generator, is considered for $1$D semigroups generated by the Schr\"odinger operators $-\partial_{x}^{2}+x^{2 }+\im \gamma f(x)$ with a complex potential.  For $f$ of the special form $f(x)=1/(1+|x|^{\kappa})$, it is shown using complex dilations that the real part of eigenvalues of the operator are larger than a constant times $ |\gamma|^{2/(\kappa+2)}$.
\end{abstract}
\maketitle
\section{Introduction}
The purpose of this note is to take up up the following  problem suggested by C.\ Villani \cite{Villani} in 
the proceedings of the last ICM, quoted here (with notation slightly
changed):
\begin{quotation}
  Identify sufficient conditions on $f: \R \ra \R$ so that the real
  parts of the eigenvalues of
  $$ L_\gamma : \psi \ \mapsto \ ( - \partial_x^2 \psi + x^2 \psi -
  \psi) + \im \gamma f \psi $$
  as on operator on $L^2$ go to infinity as $|\gamma| \ra \infty$
  and estimate the rate of divergence.
\end{quotation}
The stated problem is a model  for the phenomenon of "hypocoercivity," with the specific form  motivated by applications described in \cite{Villani} and outlined briefly below.

It turns out that a sufficient condition on $f$ for divergence of the real parts of the eigenvalues may be obtained using a general result due to Constantin, Kiselev, Ryzhik and Zlato\v{s} \cite[Theorem 1.4]{CKRZ}.  Specifically, we will show below that:
\begin{thm}\label{thm:divergence} The real parts of the eigenvalues of $L_{\gamma}$ diverge to $\infty$ as $\gamma \ra \infty$ if $\{ x : f(x) =t\}$ is essentially nowhere dense for each $t \in \R$. 
\end{thm}
\begin{rems} 1) A set $S$ is \emph{essentially nowhere dense} if $S = S' \cup N$ where $S'$ is no-where dense and $N$ has Lebesgue measure zero.  2) Since the operators $L_\gamma$ are unbounded, we should specify
their domains.  Throughout we will consider bounded $f$, so the natural choice is to consider $L_\gamma$ on the domain of self-adjointness for the real part $L_0$:
$$\cu D = \set{ \psi : \int |\partial_x \psi|^2 + x^2 |\psi|^2 <
\infty}.$$  It is then not hard to show that $L_{\gamma}$ is the generator of a contractive semigroup and that $L_\gamma^{-1}$ is compact. Thus the spectrum $\sigma(L_\gamma)$ consists only of eigenvalues.
2) In fact the ``if'' of the theorem could be replaced by an ``if an only if'' \cite{Z} -- a proof of the reverse implication would follow, for example, from the method used recently to derive a similar result in the context of diffusion with drift on a compact manifold \cite{FHPS}.
\end{rems}

However, the question of estimating the \emph{rate} of divergence seems to be much more delicate.  In particular, it is likely that the rate will depend on the specific form of $f$ and thus we may not expect such a general result.  Nonetheless, as we will show below, a number of interesting examples with analytic $f$ may be analyzed with the technique of complex dilations, yielding, in one case at least, a very good estimate on the rate. 

The method of complex dilations, as applied below, essentially require that $f$ be analytic, possibly with a branch point. Furthermore, as the proof will show, the asymptotic growth of eigenvalues depends essentially on $f$ only in the neighborhood of a critical point.  If the potential has multiple critical points one may expect each of these to contribute.  For the sake of clarity, we will not try to formulate a very general result, but rather focus on a function $f$ with a single critical point to see how the behavior of the spectrum depends on the behavior of $f$ near the critical point.  Specifically, we will consider
\begin{equation}f(x)= \frac{1}{1+|x|^{\kappa}} \label{eq:formof}\end{equation}
for arbitrary $\kappa >0$.
 The main result of this paper is
\begin{thm} \label{thm:main} Fix $\kappa >0$ and let $f(x)$ be as in \eqref{eq:formof}. Then there is a positive
  constant $C_{\kappa}$ such that all eigenvalues of $L_\gamma$ lie in
  the half plane $$\set{z \ : \ \Re z \ge C_{\kappa} |\gamma|^{2/(k+2)}}.$$
\end{thm}
\begin{rems} 1) Since the spectrum $\sigma(L_\gamma)$ consists only of eigenvalues, the theorem is  succinctly stated as the inclusion $\sigma(L_\gamma) \subset \set{z : \Re(z) \ge C
|\gamma|^{2/(\kappa+2)}}$. 2) The proof relies on a quadratic form estimate which only gives lower bounds.  However, in \cite{Villani} it is mentioned that numerical calculations by Thiery Gallay for $\kappa=2$  suggest that the obtained rate $\gamma^{1/2}$ is sharp.  This is to be contrasted with the rate $\gamma^{1/4}$ provided by the commutator methods of Villani \cite{Villani}.
\end{rems}

In \cite{Villani}, this problem is motivated as a simplified version of a
spectral problem that arises in the linear stability analysis of a
$2D$ hydrodynamic equation. In that context, one is mainly
interested in understanding the $t \ra \infty$ asymptotic behavior
of solutions to the the initial value problem
\begin{equation}\label{eq:ivp}
\partial_t \psi(x,t) \ = \ \partial_x^2 \psi(x,t) - x^2 \psi(x,t) -
\im \gamma \frac{1}{1 + |x|^\kappa} \psi(x,t), \quad \psi(x,t) =
\psi_0(x),
\end{equation}
The semi-group  $S_t = \e^{-t L_\gamma}$ generates solutions to
\eqref{eq:ivp} via $\psi(x,t) = S_t \psi_0(x)$. It follows that the
large $t$ asymptotic behavior is governed by the eigenvalue of
$L_\gamma$ with smallest real part. Indeed,
$$
  \norm{\e^{-t L_\gamma} } \ \le \ C_{\gamma,\alpha} \ \e^{- \alpha t }
  $$
for any $\alpha < \inf
  \Re(\sigma(L_\gamma)) $,
where for $A \subset \C$ we let $\Re(A) = \set{\Re(z) \ : \ z \in
A}$. Hence, the conclusion of the theorem implies that
\begin{equation}\label{eq:semigroup}
  \int_{\R} |\psi(x,t)|^2 \di x \ \le \ C_{\gamma} \ \e^{-C \gamma^{2/(\kappa+2)}
  t},
\end{equation}
for solutions to \eqref{eq:ivp}. 

The estimate \eqref{eq:semigroup} is quite striking as, in some sense, the convergence of $\psi \ra 0$ is driven entirely by
the ``dissipative'' term $\partial_x^2 \psi - x^2 \psi$ on the
right. Indeed, without the dissipative term, the solution would be
$$\psi(x,t) \ = \ \e^{-\im \gamma
 \frac{ t}{1+|x|^\kappa}} \psi_0(x), $$ which has constant $L^2$ norm. If the
$L^2$ conserving term $-\im \gamma/(1+|x|^\kappa) \psi$ is absent then the $L^2$ norm goes to zero, but at the modest rate $\e^{-t}$ given by finding the smallest eigenvalue of $L_0$. However the combination of the ``conservative'' and ``dissipative'' terms results in much faster convergence to zero. This phenomenon has been dubbed ``hypocoercivity'' by Villani, in analogy with H\"ormander's notion of ``hypoelipticity.''

From another point of view, this is not so surprising. After all, those of us who add milk to our tea  know to stir the cup of tea after adding milk to speed up the mixing of milk and tea. More or less this is the same phenomenon as what is described in eq.~\ref{eq:ivp}.   After all, on it's own the milk would eventually diffuse through the cup of tea.  However, we may greatly decrease the time to acheive diffusion by stirring, which essentially adds a convective term to the diffusion equation for the density of tea.  This convection on it's own, in an imaginary liquid with no dissipation, would only set the milk and tea in ever circulating motion -- it is "conservative"!  Together the dissipative and conservative motions combine into a flow which mixes much faster.

In recent years the mathematical analysis of hypocoercivity and related problems has been the subject of much research.  Constantin, Kiselev, Ryzhik and Zlato\v{s} \cite{CKRZ} have analyzed the phenomenon from an abstract functional analytic perspective summarized in the following theorem:
\begin{thm}[Theorem 11.4 of \cite{CKRZ}]\label{thm:CKRZ}  Let $H_{0}$ be a self-adjoint, positive, unbounded operator with discrete spectrum and let $F$ be self-adjoint and satisfy
\begin{equation*}
\norm{F \psi}^{2} \ \le \ C \ip{\psi, H_{0} \psi} \text{ and } \ip{\e^{\im F t } \psi, H_{0} \e^{\im F t} \psi} \ \le \ B(t)  \ip{\psi, H_{0} \psi} 
\end{equation*}
with constants $C$ and $B(t)$ independent of $\psi$ and $B(t) \in L^{2}_{\mathrm{loc}}(0,\infty)$.  Then for any $\gamma > 0$ the operator $H_{0} + \im \gamma F$ generates a contractive semi-group and the following are equivalent
\begin{itemize}
\item For any $t > 0$
\begin{equation}\label{eq:hypocoercive}
\lim_{\gamma \ra \infty }\norm{\e^{-t (H_{0} + \im \gamma F)}} = 0 .
\end{equation}
\item The operator $F$ has no eigenvectors in $\cu{Q}(H_{0}) = $ form domain of $H_{0}$.
\end{itemize}
\end{thm}

Using Thm.\ \ref{thm:CKRZ} we may easily prove Thm.\ \ref{thm:divergence}. Note that any normalized eigenvector $\psi$ satisfies
\begin{equation*} \norm{\e^{-t (H_{0} + \im \gamma F)} \psi} = \e^{-t \Re \lambda}
\end{equation*}
with $\lambda$ the corresponding eigenvalue.  It follows that all eigenvalues satisfy
\begin{equation}\label{eq:eigbound}
\Re \lambda \ge - \frac{1}{t} \ln \norm{\e^{-t (H_{0} + \im \gamma F)}}. 
\end{equation}
Thus \eqref{eq:hypocoercive} implies that $\lim_{\gamma \ra \infty} \inf \Re \sigma(H_{0}+ \gamma F) = \infty$. This result applies in the present context, with the operator $H_{0} = - \partial_{x}^{2} + x^{2}$ and $F = f(x)$. The form domain of $H_{0}$ is the  space
\begin{equation}\label{eq:H1}
\cu{H}_{1} = \{ \psi : \partial_{x} \psi \text{ and } x \psi \in L^{2}(\R) \}.
\end{equation}
An eigenvector of $F$ is a function supported on a level set of $f$.  Thus  $\norm{\e^{-t (H_{0} + \im \gamma F)} } \ra 0$
as $\gamma \ra \infty$ if and only if the level sets of $f$ support no (non-trivial) $\cu{H}_{1}$ eigenfunctions.  Since $\cu{H}_{1}\subset C(\R)$ Theorem \ref{thm:divergence} follows.

However, \eqref{eq:eigbound} and \eqref{eq:hypocoercive} do not provide quantitative information on the rate of divergence.   To find this we must estimate, for finite $\gamma$, the location of the eigenvalues.   Before proceeding, it is instructive to ask, in general, ``how do we estimate the location of spectrum?'' When one is available, a variational principle
relating eigenvalues to extrema of a quadratic form is one of the most effective tools. For example, consider the self-adjoint eigenvalue problem
\begin{equation}\label{eq:basiceigprob}
  -\partial_x^2 \psi(x) + V(x) \psi(x) \ = \ \lambda \psi(x) ,
\end{equation}
with real valued $V$ that diverges to $+\infty$ as $x \ra \pm \infty$.
The smallest $\lambda \in \R$  such that a solution $\psi \in
L^2(\R)$ may be found is called the \emph{ground state eigenvalue} and satisfies a variational
principle
\begin{equation}\label{eq:basicvarprinc}
  \lambda_0 \ = \ \inf \left \{ \left . \int_{-\infty}^\infty
 |\partial_x \psi(x)|^2 + V(x) |\psi(x)|^2  \di x \ \right | \
  \int_{-\infty}^\infty |\psi(x)|^2 \di x = 1 \right \}.
\end{equation}
Furthermore a minimizing wave function $\psi$ satisfies the
eigenvalue equation \eqref{eq:basiceigprob}. The min-max principle
generalizes this to higher eigenvalues.

It would be hard to overstate the utility of
\eqref{eq:basicvarprinc}. To find an upper bound on $\lambda_0$ we
just compute the energy
$  \cu{E}(\psi)  \ = \ \int_{-\infty}^\infty
 |\partial_x \psi(x)|^2 + V(x) |\psi(x)|^2  \di x$
of any normalized function $\psi \in L^2(\R)$. To find a \textit{good}
upper bound we ``simply'' make a suitably clever choice of $\psi$.
Finding a lower bound is often more involved, requiring a uniform
estimate from below on $\cu{E}(\psi)$, but even that is often
possible, particularly in an asymptotic limit (see, e.g., Simon's paper on semiclassical analysis \cite{Simon}). In way of contrast, an exact analysis of the eigenvalue problem \eqref{eq:basiceigprob} is  possible only for a few explicitly solvable examples.

However, for the operators $L_\gamma$ considered here \textemdash \
which are not self-adjoint \textemdash \  a direct variational
approach provides little insight. It remains true that any
eigenvalue $\lambda$ satisfies
\begin{equation}\label{eq:nonsavarprinc} \Re \lambda  \ \ge \ \inf
\set{ \Re \ip{\psi, L_\gamma \psi} \ : \ \int |\psi(x)|^2 = 1 },
\end{equation}
since $\Re \lambda \ = \ \Re \ip{\psi_\lambda, L_\gamma
\psi_\lambda}$ with $\psi_\lambda$ the corresponding normalized
eigenfunction. However, the infimum on the r.h.s.\ is insensitive to
the imaginary term $\im \gamma/(1+ |x|^\kappa)^m$ in the operator, and
thus is independent of $\gamma$.  So all we learn is that
$$ \Re \lambda \ \ge \ \inf \set{ \Re \ip{\psi, L_0 \psi} \ : \ \int
|\psi(x)|^2 = 1 } \ = \ 1 ,$$
using the well known explicit diagonalization of  $L_0  =
- \partial_x^2 + x^2$.

The cornerstone of the variational
approach is a relation between numerical range and spectrum valid
only for normal operators.\footnote{Recall that $L$ is normal if it commutes with
it's adjoint $L^\dagger$. Lack of self-adjointness is not really the problem here. The variational approach would work in principle if $L_\gamma$ were
normal.} The \emph{numerical range} of a linear
operator $L$ on a Hilbert space $\cu{H}$ refers to the set
\begin{equation*}
\num(L) \ = \ \set{\left . \ip{\psi, L \psi}_{\cu{H}} \,  \right |
\, \psi \in \cu{H} \text{ and } \norm{\psi}_{\cu{H}} = 1 }.
\end{equation*}
If $L$ is normal then we have
\begin{equation}
\label{eq:normalspectrum=range} \cch(\sigma(L)) \ = \ \cl
(\num(L)),
\end{equation}
where $\cch$ denotes the ``closed convex hull'' and $\cl$ denotes topological closure.
This may be seen using the spectral theorem, since it is elementary
for multiplication operators on $L^2$ spaces. In particular, one has the following
variational principle for normal operators:
\textit{extreme points of the closed numerical range are in the
spectrum}.

For a general closed operator $L$ we do not have \eqref{eq:normalspectrum=range}. However, since the \emph{point spectrum} of an operator $L$ clearly falls in $\num(L)$ we have 
\begin{equation}\label{eq:spectruminrange0}
  \sigma(L) \  \subset \ \cl (\num(L)), \end{equation}
whenever $L$ has compact resolvent --- as do the operators $L_{\gamma}$ considered here.\footnote{Eq.\ \eqref{eq:spectruminrange0} also holds whenever $L$ is bounded, but may fail for an unbounded closed operator. For example consider $M\psi=\psi'$ on the space $H^{1}_{0}(0,1)$ of functions vanishing at $0$ and $1$ but with $L^{2}$ derivative.  Then $M$ is closed and $\num(M)=\im \R$, but $\sigma(M) =\C$ since $\ran (M-zI)$ is perpendicular to $e^{\overline{z}x}$ for any $z \in \C$. }
In fact, since $\num(L)$ is convex \cite[Theorem V.3.1]{Kato}, we have
\begin{equation}
\label{eq:spectruminrange}
\cch(\sigma(L)) \ \subset \ \cl (\num(L))
\end{equation}
for $L$ with compact resolvent.
However, there does not appear to be a relation between spectrum and numerical range valid for all operators $L$ with compact resolvent beyond the inclusion
\eqref{eq:spectruminrange}.  Furthermore one knows from examples that the spectrum may lie arbitrarily far from the extreme points of the numerical range. For example, the numerical range of a $2 \times 2$ matrix is an ellipse with foci at the eigenvalues, but the radii of the ellipse may be arbitrarily large.

One quite natural idea is to exploit the invariance of $\sigma(L)$ under the map $L \mapsto T^{-1}L T$ where $T$ is bounded with bounded inverse.  This map fixes the numerical range only if $T$ is unitary, so we may gain something by this procedure.  Indeed, if we
were lucky enough to find $T$ so that $T^{-1} L T$ were normal, we would have a variational principle for the eigenvalues of $L$ using the quadratic form $\ip{\psi,T^{-1}LT}$.  Even if $T^{-1} L T$ is not normal, we may hope to use $T$ to bring 
the extreme points of the numerical range closer to the spectrum of $L$.  In finite dimension, this procedure will work, in principle, for any
$L$. Indeed, if $\dim(\cu{H}) < \infty$ then
\begin{equation}\label{eq:extendedvarprinc}
\cch (\sigma(L) ) \ = \ \bigcap_{T \in \cu{GL}(\cu{H})}
\cl(\num(T^{-1} L T)).
\end{equation}
In fact, using the Jordan canonical form one can show a bit more:
for any $\epsilon > 0$ there is $T_\epsilon \in \cu{GL}(\cu{H})$
such that $\cl(\num(T_\epsilon^{-1} L T_\epsilon))$ is contained in
the $\epsilon$ neighborhood of $\cch(\sigma(L))$.  It is not
clear if \eqref{eq:extendedvarprinc} holds for general
$L$ with compact resolvent when $\dim(\cu{H}) = \infty$.  Anyway the answer to this question may not be so relevant, as in practice it is rather difficult to
produce an effective conjugating operator without already knowing the Jordan form of $L$. 

So how do we locate the spectrum of $L_\gamma$? Let us first
consider an heuristic approach, that  is very
close in spirit to the rigorous method applied below. Let $\lambda$
be an eigenvalue, $L_\gamma \psi = \lambda \psi$, with eigenfunction
$\psi$.   Suppose that $\psi$ has a holomorphic extension $\psi(z)$ for
$z \in C_\beta := \{x \e^{\im \theta} \ : \ x \in \R \text{ and }
|\theta| < \beta\}$. By rotating the contour on which we evaluate
the eigenfunction equation we find that
$$ - \e^{- 2 \im \theta} \partial_x^2 \phi(x) + \e^{2 \im \theta}x^2 \phi(x)
+ \im \gamma \frac{1}{1 + \e^{\kappa \im \theta} |x|^{\kappa }} \phi(x) =
\lambda \phi(x),$$ where $\phi(x) = \psi(\e^{\im \theta} x)$ for
some $\theta < \beta$. (We have extended $|x|^\kappa$ to the holomorphic
function $(z^2)^{\frac{\kappa}{2}}$, single valued on $C_\beta$ with a
branch singularity at $0$ provided $\beta$ is not too large.) Assuming the
complex rotated  function $\phi$ is square integrable, we discover
that $\lambda$ is an eigenvalue of
\begin{equation} \label{eq:dilatedL} L_\gamma^{(\im \theta)} \ = \ - \e^{- 2 \im \theta} \partial_x^2  + \e^{2 \im \theta}x^2
+ \im \gamma \frac{1}{1 + \e^{\kappa \im \theta} |x|^{\kappa}}.\end{equation} In
particular, applying \eqref{eq:nonsavarprinc} with $L_\gamma^{(\im
\theta)}$ in place of $L_\gamma$, we find $\Re \lambda \ge \inf
\sigma(\Re L_\gamma^{(\im \theta)})$ where
$$
\Re L_\gamma^{(\im \theta)}\ = \  \frac{1}{2} \bigl [ L_\gamma^{(\im
\theta)} + {L_\gamma^{(\im \theta)}}^\dagger \bigr ] \ = \ \cos(2
\theta) \left \{ -
\partial_x^2 + x^2 + \gamma \frac{ \sin(\kappa \theta)}{\cos(2 \theta)}
\frac{ |x|^\kappa }{\abs{1 + \e^{\kappa \im
\theta} |x|^\kappa}^{2}} \right \}.
$$
The operator in curly brackets is a Schr\"odinger operator with
(real valued) potential
$$V(x) \ = \ x^2 + \gamma \frac{ \sin(\kappa \theta)}{\cos(2 \theta)}
\frac{ |x|^\kappa }{\abs{1 + \e^{\kappa \im
\theta} |x|^\kappa}^{2}}.$$ If $\theta$ is sufficiently small then $V(x)
> x^2$ and near the origin
$$V(x) \ \sim \ x^2 +  \gamma \frac{\sin(\kappa \theta)}{\cos(2 \theta)}
|x|^\kappa .$$ Applying ideas from semi-classical analysis, as in \cite{Simon}, one
concludes for large $\gamma$ that the ground state of $\Re L_\gamma^{(\im
\theta)}$ is an approximate ground state of $$\cos(2 \theta) \left \{
-\partial_x^2 + \gamma \frac{\sin(\kappa \theta)}{\cos(2 \theta)} |x|^\kappa
\right \},$$ whose ground state eigenvalue is seen to be proportional to
$\gamma^{2/(\kappa+2)}$ by scaling.

There are two deficiencies with the above argument.  First we do not
know that the wave function is analytic. Second, even if $\psi$ were
analytic, there would be no \emph{a priori} reason to believe that $\phi
\in L^2$. A reader familiar with the technique of "complex scaling" from scattering theory for (self-adjoint) Schr\"odinger operators  will recognize the way out. We simply \emph{ignore} these problems, focusing our attention instead on the spectral analysis of the following analytic family of operators
\begin{equation}\label{analyfam}
L_{\gamma}^{(w)}  \ = \ -  \e^{- 2 w} \partial_x^2 + \e^{2 w} x^2
+ \im \gamma \frac{1}{1 + \e^{\kappa w} \abs{x}^{\kappa}}
\end{equation}
with $w=u+\im \theta$ a complex parameter in a strip $\{ | \theta | <
\delta \}$ with  $ \delta$ sufficiently small. For real $w=u$, with $\theta=0$, we have
\begin{equation}\label{eq:dilateL0} L_\gamma^{(u)} =
T_u L_\gamma T_u^\dagger,$$ where $T_u$ is the unitary dilation operator,
$$ T_u \psi(x) \ = \ \e^{\frac{1}{2} u} \psi(\e^{u x}).
\end{equation}
However, for complex $w$, $L_{\gamma}^{(w)}$ is not obtained from $L_\gamma$ by conjugation, but rather by analytic continuation from the operator valued map $u \mapsto L_\gamma^{(u)}$.  With this set-up, we will prove Theorem \ref{thm:main} in two steps:
\begin{enumerate}
\item From analyticity we will show that the spectrum of $L_\gamma^{(w)}$ is independent of $w$ as $w$ varies in the strip $\{|\Im w| < \delta \}.$
\item Using semi-classical analysis of $\Re L_\gamma^{(i \theta)}$ for $\theta \neq 0$ we will obtain an effective estimate on the real part of the lowest eigenvalue.
\end{enumerate}


\section{Proof}
For each $\alpha >0$ let $S_\alpha = \set{w\in \C \ : \ |\Im w|< \alpha  } .$
The first step of the argument is to show that $w \mapsto \left [ L_\gamma^{(w)}\right ]^{-1}$ is a compact operator valued analytic map on a strip $S_{\alpha_\kappa}$.
\begin{lem}Let
$$\alpha_\kappa = \pi \min \left ( \frac{1}{2\kappa},\frac{1}{2} \right ).$$
For each $w\in S_{\alpha_\kappa}$ the operator $L_\gamma^{(w)} $ has compact resolvent and
$$ \sigma(L_\gamma^{(w)}) = \sigma (L_\gamma).$$
\end{lem}
\begin{proof}  We will analyze $L_\gamma^{(w)}$ through the associated quadratic forms 
\begin{equation}\label{eq:forms}
Q_\gamma^{(w)}  (\psi) =e^{-2w} \norm{\partial_x\psi}_{L^2}^2 +e^{2w} \norm{x\psi}_{L^2}^2 + \im \alpha \int_{-\infty}^\infty \frac{1}{ 1+e^{\kappa w} \abs{x}^{\kappa}} |\psi(x)|^2 \di x,
\end{equation}
defined on the form domain $\cu{H}_1$ (see \eqref{eq:H1}). Our goal is to show, in the terminology of Kato \cite{Kato}, that $w \mapsto Q_\gamma^{(w)}$ is a holomorphic family of type (a), which is to say
\begin{enumerate}
\item $w\mapsto Q_\gamma^{(w)}(\psi)$ is a holomorphic map for each $\psi \in \cu{H}_1$, and
\item for fixed $w$ the form $Q_\gamma^{(w)}(\cdot)$ is sectorial and closed on $\cu{H}_1$.
\end{enumerate}
Once we have shown (1) and (2), it follows by definition that $L^{(w)}_\gamma$ is a ``holomorphic family of type (B)'' and thus, by Theorem 4.3 of \cite{Kato} that $L_\gamma^{(w)}$ either has compact resolvent for all $w$ or for  no $w$.  Since $L_\gamma=L_\gamma^{(0)}$ has compact resolvent it follows that $L_\gamma^{(w)}$ has compact resolvent for all $w$. 

Take (1) and (2) for granted for the moment.  Then $w \mapsto L_\gamma^{(w)}$ is a holomorphic family with compact resolvent.  By \cite[Theorem VII.1.9]{Kato}, we have the further dichotomy that $\lambda \in \C$ is an eigenvalue of $L_\gamma^{(w)}$ either for $w$ in a discrete set or for all $w\in S_{\alpha_k}$.  We may rule out the possibility of a discrete set by noting that 
$$L_\gamma^{(u+\im \theta)}=T_u^\dagger L_\gamma^{(\im \theta)} T_u,$$
where $T_u$ is the unitary dilation operator \eqref{eq:dilateL0}.  That is, the family $u \mapsto L^{(u+\im \theta)}$ is related by unitary conjugations and is thus isospectral.  Thus $L_\gamma^{(w)}$ is isospectral as claimed and the lemma is proved once we verify (1) and (2). 

Turning now to (1) and (2), note that (1) follows easily from the explicit expression for $Q_\gamma^{(w)}$.  (Since $\Re e^{\kappa w}>0$ on $S_{\alpha_\kappa}$ the integrand in the third term on the right hand side of \eqref{eq:forms} cannot develop a singularity.)   To prove (2), namely sectoriality, note that
\begin{equation*} Q_\gamma^{(u+\im \theta)}= Q_\gamma ^{(\im \theta)} \circ T_u,
\end{equation*}
where the dilation $T_u$ clearly maps $\cu{H}_1$ onto $\cu{H}_1$. Thus $Q_\gamma^{(u+\im \theta)}$ is closed and sectorial if and only if $Q_\gamma^{(\im \theta)}$ is closed and sectorial.

To show that $Q_\gamma^{(\im \theta)}$ is closed and sectorial, note that
$$
\Re Q_\gamma^{(\im \theta)}(\psi) = \cos (2 \theta) \left \{ \norm{\partial_x \psi}_{L^2}^2 +   \norm{x \psi}_{L^2}^2 \right \} - \gamma \int_{-\infty}^\infty \Im  \frac{1}{1+e^{\im \kappa \theta} |x|^\kappa} \abs{\psi(x)}^2 \di x .$$
Thus, since $\cos (\kappa \theta) >0$ for $|\theta| <\alpha_\kappa$, we have
\begin{equation}\label{eq:ulbounds} \cos(2\theta) Q_0 (\psi) - |\gamma| \norm{\psi}_{L^2}^2 \le \Re Q_\gamma^{(\im \theta)}(\psi) \le \cos (2\theta)Q_0(\psi)+|\gamma| \norm{\psi}_{L^2}^2,\end{equation}
where
$$Q_0(\psi)=\norm{\partial_x \psi}_{L^2}^2 +   \norm{x \psi}_{L^2}^2  .$$  
Now
$$\Im Q_\gamma^{(\im \theta)}(\psi) = \sin (2 \theta) \left \{ - \norm{\partial_x \psi}_{L^2}^2 +   \norm{x \psi}_{L^2}^2 \right \} + \gamma \int_{-\infty}^\infty \Re \frac{1}{1+e^{\im \kappa \theta} |x|^\kappa}\abs{\psi(x)}^2 \di x .$$
Thus, by \eqref{eq:ulbounds},
$$\abs{\Im Q_\gamma^{(\im \theta)}(\psi)} \le \sin(2\theta) Q_0(\psi) +|\gamma| \norm{\psi}_{L^2}^2\le \tan(2\theta) \Re Q_\gamma^{(\im \theta)}(\psi) +|\gamma|  \left ( 1 +\frac{1}{\cos(2\theta)} \right ) \norm{\psi}_{L^2}^2,$$
which is to say $Q_\gamma^{(\im \theta)}$ is sectorial.  Furthermore, since
$Q_0$ is a closed form with domain $\cu{H}_1$ it follows from \eqref{eq:ulbounds} that $\Re Q_\gamma^{(\im \theta)}$ and hence $Q_\gamma^{(\im \theta)}$ is closed on $\cu{H}_1$. 
\end{proof}

\subsection{Analysis of $\Re L_\gamma^{(\im \theta)}$ }  
Since $L_\gamma$ and $L_\gamma^{(\im \theta)}$ are isospectral, if $\lambda$ is an eigenvalue of $L_\gamma$ then
\begin{equation*}
\mathrm{Re} \lambda \ \ge \  \cos(2 \theta) \min \sigma(H_{\gamma}(\theta) ) ,
\end{equation*}
where $H_{\gamma}(\theta)$ is the real part of $\frac{1}{ \cos(2 \theta)}
L_\gamma^ {(\im \theta)}$. That is,
\begin{equation*}
H_\gamma(\theta) \ = \  -  \partial_x^2 +  x^2 + \gamma
\frac{\sin(\kappa \theta)}{\cos(2 \theta)} \frac{|x|^\kappa }{ |1 + \e^{ \im \kappa
\theta} |x|^{\kappa} |^{2}} .
\end{equation*}
Note that $H_\gamma(\theta)=H_{-\gamma}(-\theta)$.  Hence without loss of generality we may take $\gamma >0$.
For the rest of the proof, let us fix some $\theta \in (0,\alpha_{\kappa})$. Theorem \ref{thm:main} will follow once we show that the ground state eigenvalue $\lambda_0 =\lambda_0(\theta)$ of $H_\gamma(\theta)$ satisfies an estimate
\begin{equation}\label{eq:asymptoticmain}\lambda_0 \ge C \gamma^{\frac{2}{\kappa+2}}.
\end{equation}
(In the end one could try to optimize over the choice of $\theta$, however  this would affect only the proportionally constant not the rate of divergence of the eigenvalue.)

The Schr\"odinger operator $H_\gamma(\theta)$ is of the form
$$-\partial_x ^2 + V(x),$$
where 
$$V(x)= x^2 + \gamma \frac{\sin(\kappa \theta) }{\cos(2\theta)}\frac{ |x|^\kappa}{ |1 + \e^{ \im \kappa
\theta} |x|^{\kappa} |^{2}} .$$
Thus $V(x)$ has a global minimum at $x=0$, with 
$$V(x)= \gamma \frac{ \sin(k\theta) }{\cos (2 \theta)} |x|^\kappa + x^2 + o(|x|^\kappa), \quad |x| \rightarrow 0.$$ Eq.\ \eqref{eq:asymptoticmain}
follows if we can show that the ground state eigenvalue of $H_\gamma(\theta)$ is asymptotic to the ground state eigenvalue of the anharmonic oscillator
\begin{equation} K_\alpha = -\partial_x^2 +\alpha |x|^\kappa \label{eq:anharmonic}\end{equation}
with $\alpha = \gamma \frac{\sin(k\theta)}{\cos(2\theta)}$.  Indeed, scaling shows that the ground state of $K_\alpha$ satisfies
$\lambda_0(K_\alpha)= \lambda_{0}(K_{1}) \alpha^{\frac{2}{2+\kappa}}$.

Let $\Phi $ be a fixed  $C^{2}$ ``cut-off'' function, with
\begin{enumerate}
\item $0\le\Phi(x)\le1$ for all $x$,
\item $\Phi(x)=1$ for $|x|\le 1$,
\item $\Phi(x) =0$ for $|x| \ge 2$.
\end{enumerate}
Fix $\nu > 0 $, to be chosen below, and let $$J(x) = \Phi(
\alpha^\nu x), \quad \text{with }\alpha = \gamma \frac{\sin(k\theta)}{\cos(2\theta)}.$$
Define $\widetilde J(x) $ so that $\widetilde J(x)^2 + J(x)^2 = 1$, that is $\widetilde J(x) =  \sqrt{1 -
J(x)^2}$.  By the IMS localization formula (see \cite[Theorem 3.2]{CFKS}),
\begin{equation}
H_\gamma(\theta) \ = \ \widetilde J H_\gamma(\theta) \widetilde J + J H_\gamma(\theta) J
- (\partial \widetilde J)^2 - (\partial J)^2. \label{eq:IMS}
\end{equation}

We will estimate each of the four operators on the right hand side of \eqref{eq:IMS} separately. To estimate $\widetilde J H(\theta) \widetilde J$, 
we simply drop the kinetic term $-\partial_x^2$, to obtain
\begin{equation*} \widetilde J H_\gamma(\theta) \widetilde J  \ \ge \
 \inf_{|x| >  \alpha^{-\nu}} \left [ x^2 + \alpha  \frac{ |x|^\kappa }{|1 + \e^{ \im \kappa \theta} |x|^{\kappa} |^{2}} \right ] \widetilde J^2 \
\ge \  \inf_{y > \alpha^{-\kappa \nu}}    \left [ y^{2/\kappa} +
\alpha  \frac{y}{(1+ y)^2} \right ]
\widetilde J^2.
\end{equation*}
To minimize $g(y)=y^{2/\kappa} +
\alpha  \frac{y}{(1+ y)^2}$ over $y  > \alpha^{-\kappa \nu}$ we must compare
$$ g(\alpha^{-\kappa \nu})=\alpha^{-2 \nu} +\alpha \frac{\alpha^{-\kappa \nu}}{(1+\alpha^{-\kappa \nu})^{2}} \sim C \alpha^{1-\kappa \nu}$$
with the value of $g$ at the stationary points, which solve 
$$\frac{2}{\kappa \alpha} y^{\frac{2}{\kappa}-1}(1 + y)^{3}+ 1= y.$$
There are two solutions to this equation, $y_{1} \sim 1$ and $y_{2} \sim C \alpha^{\frac{\kappa}{2+\kappa}}$
as $\alpha \rightarrow \infty$.  Evaluating $g$ at these points we obtain
$$ g(y_{1})\sim \frac{1}{2} \alpha \quad \text{and} \quad g(y_{2}) \sim C \alpha^{\frac{2}{2+\kappa}}.$$
 Thus, provided we take $\nu \le \frac{2}{2+\kappa}$, we conclude that
$$ \widetilde J H_{\gamma}(\theta) \widetilde J \ \ge \ C  \alpha^{\frac{2}{2+\kappa}} \widetilde J^2$$
as $\alpha \rightarrow \infty$.

To estimate $J H_{\gamma}(\theta) J$ we compare to the anharmonic
oscillator \eqref{eq:anharmonic}:
\begin{multline*} JH_{\gamma}(\theta) J - JK_{\alpha} J \  =  \ \left [ x^{2} + \alpha |x|^{\kappa} \left [ \frac{1}{|1+e^{\im \kappa \theta}|x|^{\kappa}|^{2}} -1 \right ] \right ] J^{2} \\
\ge -\alpha \sup_{ |x| \le  2 \alpha^{-\nu}}  \left [ |x|^{\kappa} \left [ 1 - \frac{1}{|1+|x|^{\kappa}|^{2}} \right ] \right ] J^2 \
\sim \ -C \alpha^{1-2 \kappa \nu} J^{2}
\end{multline*}
as $\alpha \rightarrow \infty$.  Putting this together with the estimate for $\widetilde J H_{\gamma}(\theta) \widetilde \gamma$  and the estimate $K_{\alpha} \ge C \alpha^{\frac{2}{2+\kappa}}$, we see that
\begin{multline*}
H_{\gamma}(\theta) \ = \  J K_{\alpha}J + J (H_{\gamma}(\theta) - K_{\alpha}) J + \widetilde J H_{\gamma} (\theta) \widetilde J - \left ( \partial_{x} J \right )^{2} - \left ( \partial _{x}\widetilde J \right )^{2} \\ 
\ge \ C \alpha^{\frac{2}{2+\kappa}} J^{2} + C \alpha^{1-2\kappa \nu} J^{2} + C \alpha^{\frac{2}{2+\kappa}} \widetilde J^{2} - C\alpha^{2\nu} \\
\ge \ C \alpha^{\frac{2}{2+\kappa}} - \left \{ C \alpha^{1-2\kappa \nu} J^{2} +C\alpha^{2 \nu} \right \},
\end{multline*}
since $|\partial_{x} J|, \ |\partial_{x}\widetilde J| \le C \alpha^\nu$.  To minimize the growth of the error term (in brackets), let us choose $\nu$ so that the two terms are of equal magnitude, namely $\nu= \frac{1}{2 + 2 \kappa}$, yielding
$$
H_{\gamma}(\theta) \ \ge \ C \gamma^{\frac{2}{2+\kappa}} - C \gamma^{\frac{1}{1+\kappa}},$$
and completing the proof. \qed

\subsection*{Acknowledgments}  I would like to express my gratitude for the hospitality of Tom Spencer and the Institute for Advanced Study, where I was a member from 2005-2007.  I learned of the problem considered here from Cedric Villani during a very pleasant discussion with him and Tom at the IAS in 2007.   I benefited also from discussions with Andrej Zlatos, who pointed out the relevance of Theorem 1.4 of \cite{CKRZ}.  The final version of this manuscript was prepared at the Centre Interfacultaire Bernoulli of the Ecole Polytechnic Federal de Lausanne during a visit in connection with the program Spectral and Dynamical Properties of Quantum Hamiltonians. 
This work supported in part by NSF CAREER Award DMS-08446325.


\end{document}